\documentclass[USenglish]{article}
\usepackage{etex}
\usepackage[utf8]{inputenc}
\usepackage[big,online]{dgruyter}
\setcitestyle{numbers,sort&compress}

\usepackage{a4wide,color}
\usepackage{amsmath,amsfonts,amssymb,mathrsfs,amsthm,bbm,dsfont,bm}
\usepackage{url}
\usepackage{diagrams}
\usepackage{enumitem}  
\usepackage{tikz}
\usepackage{authblk}    
\renewcommand\Affilfont{\itshape\footnotesize}

\theoremstyle{dgthn}
\newtheorem{theorem}{Theorem}[section]

\newtheorem{definition}[theorem]{Definition}
\theoremstyle{dgdef}

\usepackage[final]{showkeys}

\newenvironment{remark}%
  {\par\medbreak\refstepcounter{theorem}%
    \noindent\textbf{Remark~\thetheorem. }}%
  {\qed\par\medskip}

\makeatletter
  \newcommand{\eqnum}{\leavevmode\hfill\refstepcounter{equation}\textup{\tagform@{\theequation}}} 
\makeatother

\numberwithin{equation}{section}

\let\theoldvmacro\v
  \def\v#1{\ifmmode{\bm #1}\else{\theoldvmacro#1}\fi}

\let\oldint\int
\def\int{\oldint\nolimits}

\def\cA{\mathcal A}

\def\S{\mathcal S}

\def\Z{\mathcal Z}

\def\cL{\mathscr L}

\def\R{\mathbb R}

\DeclareMathOperator\Prob{Prob}
\def\Expectation{{\mathbb E}}

\renewcommand{\div}{\mathop{\mathrm{div}}}

\begin{document}

	\articletype{Research Article}
  \journalname{J.~Non-Equilib.~Thermodyn.}
  \journalyear{2015}
  \journalvolume{???}
  \journalissue{???}
  \startpage{1}
  \aop
  \DOI{10.1515/jnet-YYYY-XXXX}


\title{A generalization of Onsager's reciprocity relations to gradient flows with nonlinear mobility}
\runningtitle{Generalization of Onsager's reciprocity relations}

\author[1,2]{A. Mielke}
\author[3,*]{M. A. Peletier}
\author[1]{D.R.M. Renger}
\affil[1]{Weierstra\ss-Institut, Mohrenstra\ss e 39, 10117 Berlin, Germany}
\affil[2]{Institut f{\"u}r Mathematik, Humboldt-Universit{\"a}t zu Berlin, Rudower Chaussee 25, 12489 Berlin (Adlershof), Germany}
\affil[3]{Department of Mathematics and Computer Science, and Institute for Complex Molecular Systems (ICMS), Eindhoven University of Technology, P.O. Box 513, 5600 MB Eindhoven, the Netherlands}
\affil[*]{Corresponding author; m.a.peletier@tue.nl}

\abstract{Onsager's 1931 `reciprocity relations' result connects microscopic time-reversibility with a symmetry property of corresponding macroscopic evolution equations. Among the many consequences is a variational characterization of the macroscopic evolution equation as a gradient-flow, steepest-ascent, or maximal-entropy-production equation. Onsager's original theorem is limited to close-to-equilibrium situations, with a Gaussian invariant measure and a linear macroscopic evolution. In this paper we generalize this result beyond these limitations, and show how the microscopic time-reversibility leads to natural generalized symmetry conditions, which take the form of generalized gradient flows.  }
\maketitle

\noindent\textbf{Keywords:} gradient flows, generalized gradient flows, large deviations, symmetry, microscopic reversibility

\communicated{...}
\dedication{...}

\received{...}
\accepted{...}

\section{Introduction}

\subsection{Onsager's reciprocity relations}

In his two seminal papers in 1931, Lars Onsager showed how time-reversibility of a system implies certain symmetry properties of macroscopic observables of the system~\cite{Onsager31,OnsagerMachlup1953}. In modern mathematical terms the main result can be expressed as follows. 
\begin{theorem}
\label{th:Onsager}
Let $X_t$ be a Markov process in $\R^n$ with transition kernel $P_t(dx|x_0)$ and invariant measure $\mu(dx)$. Define the expectation $z_t(x_0)$ of $X_t$ given that $X_0= x_0$, 
\[
z_t(x_0) = \Expectation_{x_0} X_t = \int x P_t(dx|x_0).
\]
Assume that
\begin{enumerate}
\item $\mu$ is \emph{reversible}, i.e., for all $x$, $x_0$, and $t>0$, $\mu(dx_0)P_t(dx|x_0) = \mu(dx)P_t(dx_0|x)$;
\item \label{th:Onsager:cond-mu-Gaussian} $\mu$ is Gaussian with mean zero and covariance matrix $G$.
\item \label{th:Onsager:cond-evolution}
$t\mapsto z_t(x_0)$ satisfies the equation 
\begin{equation}
\label{eq:Onsager}
\dot z_t = -A z_t,
\end{equation}
for some nonnegative $n\times n$ matrix $A$. 
\end{enumerate}

\medskip

Then $M:= AG$ is symmetric positive definite, and if we define $\S(x)$ by $\mu(dx) = \exp( \S(x))\, dx$, then 
equation~\eqref{eq:Onsager} can be written as 
\begin{equation}
\label{eq:OnsagerS}
\dot z_t = M D\S(z_t).
\end{equation}
Here we write $D\S$ for the derivative of $\S$.
\end{theorem}

\bigskip

Note that we purposefully avoid the use of thermodynamic terms such as temperature, pressure, or energy, in order to highlight that the result is essentially \emph{mathematical} in nature (which becomes even more evident upon considering the proof as described in e.g.~\cite[Sec.~VII.4]{GrootMazur62}). In the context of thermodynamics, the `Markov process' typically is some stochastic process given by the macroscopic observables in question. 

\bigskip

The symmetry property $M^T=M$ states equality of corresponding off-diagonal elements of $M$, which by equation~\eqref{eq:OnsagerS} translates into equality of `cross-effect coefficients'. These are known as the `reciprocity relations', and are a corner stone of linear irreversible thermodynamics. In addition to the practical convenience of reducing the number of parameters, they also present a rare simplifying insight into the wild world of irreversible processes.

In its original form, however, Theorem~\ref{th:Onsager} only applies to near-Gaussian fluctuations around a given stationary state, which translates into `close to equilibrium'. For this reason, many generalizations have been proposed, e.g.~\cite{Casimir45,Ziegler58,Gyarmati70,HurleyGarrod82,Garcia-ColinRio-Correa84,GallavottiCohen95,Gallavotti96,MaesNetocny07,Seifert12,ReinaZimmer15TR}.

Despite these advances, there is still a major challenge. In Onsager's argument the macroscopic symmetry property arises from the microscopic reversibility, and it is exactly this connection that gives Onsager symmetry its force. This connection is generalized in our paper~\cite{MielkePeletierRenger14}; in the current paper we present this mathematical micro-to-macro connection from a more thermodynamic perspective. 

The central observation that allows us to move forward is that the property `symmetry of a matrix' is intrinsically linked to the linear nature of a matrix---the fact that a matrix is a linear operator. This linear nature is an obstacle to generalization to nonlinear behaviour. We therefore generalize Onsager's relations in three steps:
\begin{enumerate}
\item[I.] We first remark that symmetry and positive definiteness of the mobility $M$ {are} equivalent to the property that equation~\eqref{eq:OnsagerS} is a \emph{gradient flow}; 
\item[II.] We next generalize the concept of gradient flows (with linear mobilities) to \emph{generalized gradient flows}, in which the mobility may be a  nonlinear map;
\item[III.] We finally show that under more general conditions than Onsager's, the macroscopic evolution is such a generalized gradient flow.
\end{enumerate}
This final point is the generalization of Onsager's relations: instead of guaranteeing symmetry of some linear operator, which is equivalent to the macroscopic equation being a gradient flow, we guarantee under more general conditions that the macroscopic equation is a \emph{generalized} gradient flow. We argue in this paper that this generalized gradient-flow property is the natural generalization of Onsager symmetry, from the point of view of microscopic reversibility.

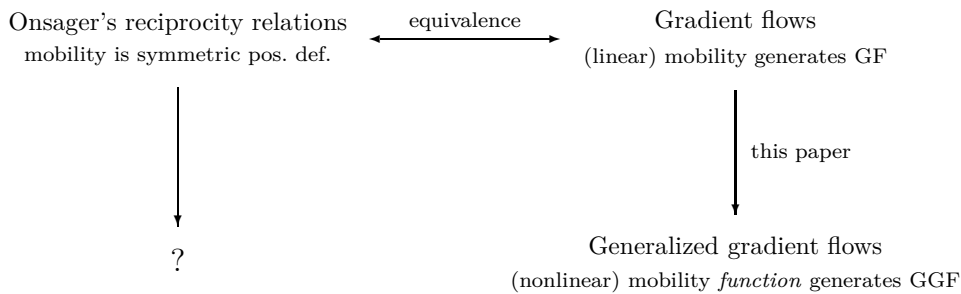
\newarrow{Dto}<--->
\noindent
\hskip-3cm
\begin{diagram}[nohug,balance,midshaft,heads=LaTeX,width=2cm,height=1.5cm,shortfall=3mm]
 \substack{\text{\normalsize Onsager's reciprocity relations}\\[\jot]\text{\footnotesize mobility is symmetric pos. def.}}&\rDto^{\text{\small equivalence}}& 
 \substack{\text{\normalsize Gradient flows}\\[2\jot] \text{\footnotesize (linear) mobility generates GF}}\\
\dTo&& \dTo_{\text{\small \ this paper}}\\           
\hbox{\large ?}& &        \substack{\text{\normalsize Generalized gradient flows}\\[\jot]\text{\footnotesize (nonlinear) mobility \emph{function} generates GGF}} \\
\end{diagram}

We walk through these three steps one by one in the sections that follow.

\bigskip

\begin{remark}
One should pay attention to the fact that in probability theory reversibility means detailed balance, whereas in thermodynamics, reversibility means that quantities such as entropy or  free energy remain constant along evolutions. In fact, probabilistic \emph{reversibility} often leads to thermodynamic \emph{irreversibility}, as is the case here.
\end{remark}

\begin{remark}
The results that we discuss in this paper are described in mathematical terms and in full detail in~\cite{MielkePeletierRenger14}. In this paper we focus on the consequences of these results for thermodynamics.
\end{remark}

\section{Step I: Definiteness is equivalent to gradient structure}

We use the language of geometry, in which the \emph{states}~$z$ are elements of a manifold $\Z$, and at each $z\in \Z$ there is a \emph{tangent plane} $T_z\Z$ consisting of all vectors tangent (thermodynamic fluxes) to $\Z$ at $z$. If $t\mapsto z_t$ is a smooth curve along the manifold $\Z$, then the time derivative $\dot z_t$ is a tangent vector, i.e. $\dot z_t\in T_{z_t}\Z$. 
Dual to the tangent plane $T_z\Z$ at any $z\in \Z$ is the \emph{cotangent plane} $T^*_z\Z$, whose elements (driving forces) are linear functionals on tangents, i.e. each $\xi\in T_z^*\Z$ is a linear mapping from  $T_z\Z$ to $\R$. We write  $\langle s,\xi \rangle$ for the dual pairing between $\xi\in T_z^*\Z$ and $s\in T_z\Z$, and we write $T\Z$ and $T^*\Z$ for the collection of all tangent and cotangent planes. 


\begin{definition}
We call the triple $(\Z,\S,M)$ a \emph{gradient system} if $\S:\Z\to\R$ and for each $z\in \Z$,
$M(z)$ is a symmetric positive definite linear mapping from $T_z^*\Z$ to $T_z\Z$.

A curve $z:\lbrack0,T\rbrack\to\Z$ is then called a \emph{solution of the gradient flow} if it solves the differential equation
\begin{equation}
\label{eq:GF}
\dot z_t = M\bigl(z_t\bigr) D\S\bigl(z_t\bigr)
\qquad\text{or equivalently} \qquad
D\S\bigl(z_t\bigr) = M(z_t)^{-1} \dot z_t.
\end{equation}
\label{def:gradient system}
\end{definition}

Although the equation~\eqref{eq:GF} is also meaningful for more general linear mappings $M$, we purposefully require symmetry and positive definiteness in order to call~\eqref{eq:GF} a gradient flow, for a number of reasons. First, the positive-definiteness forces $\S$, which in many cases one can interpret as an entropy or a negative free energy, to \emph{increase} along solutions, since 
\begin{equation}
  \frac d{dt} \S\bigl(z_t\bigr) = \langle D\S\bigl(z_t\bigr), \dot z_t\rangle = \langle D\S\bigl(z_t\bigr),M\bigl(z_t\bigr) D\S\bigl(z_t\bigr)\rangle \geq 0.
\label{eq:linear gradient flow Lyapunov}
\end{equation}
Secondly, positive-definiteness implies invertibility of $M(z)$ (there are generalizations possible where $M$ is only semidefinite). The third reason lies in the variational characterization of $\dot z$ that we discuss below. If $M(z)$ and $M(z)^{-1}$ are symmetric and positive definite, then they define norms $\frac12\xi M(z)\xi$ and $\frac12 s^T M(z)^{-1}s$ on cotangents $\xi\in T_z^*\Z$ and tangents $s\in T_z\Z$ respectively. 

Using the norm on the tangent space, the evolution can also be given a variational characterization, already recognized by Onsager, as
\begin{multline*}
\text{a curve $t\mapsto z_t$ is a solution of the gradient flow if and only for each $t$ and $z_t$,}\\
\text{$\dot z_t$ solves the maximization problem}\ 
\max_{s\in T_z\Z}\  D\S\bigl(z_t\bigr)^T s - \frac12 s^T M(z)^{-1}s.
\end{multline*}
Since the first term in the minimised expression above can be interpreted as the instantaneous production of $\S$, this evolution is often interpreted as `steepest ascent' of $\S$ (e.g.~\cite{Beretta87,MartyushevSeleznev06,Beretta14}). 
This variational characterization depends critically on the symmetry of $M$, and this is why the name `gradient system' is reserved for symmetric operators $M$.

\medskip

Note that in such a \emph{gradient flow} the mobility operator $M = M(z)$ may depend in any way on the position $z$, but the dependence on the argument $\xi = D\S(z)$ is linear: $\xi \mapsto M(z)\xi$ is a linear map. This is different from the case of nonlinear mobilities $M$, which are often written as $\xi \mapsto M(z,\xi)\xi$, in which this map depends nonlinearly on $\xi$. The \emph{generalized} gradient-flow concept that we now introduce allows for nonlinear dependence on $\xi$.


\section{Step II: Generalized gradient systems}
\label{subsec:generalized gradient flows}

We generalize to systems with nonlinear mobility $M$ by noting that a gradient flow can alternatively be formulated by defining, for $z\in\Z, s\in T_z\Z$ and $\xi\in T^*_z\Z$,
\begin{align}
\label{eq:quadratic Psi Psi*}
  \Psi(z,s):=\frac12\langle M(z)^{-1}s,s\rangle &&\text{and its convex dual}&&   \Psi^*(z,\xi):=\frac12\langle \xi,M(z)\xi\rangle.
\end{align}
Here the convex duality is in the sense of Legendre-Fenchel transforms, 
\begin{align}
  \Psi^*(z,\xi)=\sup_{s\in T_z\Z} \langle \xi,s\rangle - \Psi(z,s), &&\text{and} &&\Psi(z,s)=\sup_{\xi\in T_z^*\Z} \langle \xi,s\rangle - \Psi^*(z,\xi).
\label{eq:Psi-Psi* dual pair}
\end{align}
Equation \eqref{eq:GF} can then also be written as
\begin{equation}
\label{eq:Psi* gradient flow}
  D_s\Psi(z_t,\dot z_t) = D\S(z_t) \qquad\text{or equivalently}\qquad  \dot z_t = D_\xi\Psi^*\big(z_t,D\S(z_t)\big),
\end{equation}
which, by Legendre-Fenchel theory, is equivalent to requiring
\begin{equation}
\label{eq:Psi Psi* formulation}
  \Psi(z_t,\dot z_t) + \Psi^*\big(z_t,D\S(z_t)\big) - \langle D\S(z_t),\dot z_t\rangle=0.
\end{equation}

\emph{Generalized} gradient systems are again of the form~\eqref{eq:Psi* gradient flow} and~\eqref{eq:Psi Psi* formulation}, but where $\Psi$ and $\Psi^*$ need not be of the quadratic form~\eqref{eq:quadratic Psi Psi*}. Such formulations have a long history of study in both the thermodynamic~\cite{Edelen72,Grmela93,GrmelaOttinger97,Grmela10,Grmela12} and the mathematical literature~\cite{DeGiorgiMarinoTosques80,De-Giorgi93,ColliVisintin90,LuckhausSturzenhecker95,MielkeRossiSavare08,Mielke11a,Mielke11}, under various different names, such as Dissipation potentials, Entropy-dissipation structures, $\Psi$-$\Psi^*$-structures, and many others.
Here, $\Psi$ and $\Psi^*$ need to satisfy a number of conditions in order to yield a meaningful generalized formulation; these are summarized in the following definition.
\begin{definition}
We call the triple $(\Z,\S,\Psi)$ a \emph{generalized gradient system} if $\S:\Z\to\R$, $\Psi:T\Z\to\R$,  and for all $z\in\Z$:
\begin{enumerate}
\item $\Psi(z,\cdot)$ is convex in the second argument,
\item\label{def:GGSitem2} $\min\Psi(z,\cdot)=0$, and
\item\label{def:GGSitem3} $\Psi(z,0)=0$.
\end{enumerate}
In addition, $\Psi$ is called \emph{symmetric} if $\Psi(z,s)=\Psi(z,-s)$ for all $(z,s)\in T\Z$.

A curve $z:\lbrack0,T\rbrack\to\Z$ satisfying \eqref{eq:Psi* gradient flow} or~\eqref{eq:Psi Psi* formulation} is then called a solution of the \emph{generalized gradient flow} induced by the generalized gradient system $(\Z,\S,\Psi)$. 
\label{def:gen gradient system}
\end{definition}

With the duality~\eqref{eq:Psi-Psi* dual pair} we have the following equivalences for Conditions~\ref{def:GGSitem2}  and \ref{def:GGSitem3},
\begin{subequations}
\begin{align}
  &\min\Psi(z,\cdot)=0                 && \iff && \Psi^*(z,0)=0,            \label{eq:min Psi iff Psi*}\\
  &\Psi(z,0)=0                         && \iff && \min\Psi^*(z,\cdot)=0     \label{eq:Psi iff min Psi*},
\end{align}
\end{subequations}
so that both $\Psi\geq0$ and $\Psi^*\geq0$. This is a central part of the definition; it implies that for curves satisfying~\eqref{eq:Psi Psi* formulation}, we have the nonlinear analogue of \eqref{eq:linear gradient flow Lyapunov}:
\[
\frac{d}{dt}\S(z_t)=\langle D\S(z_t),\dot z_t\rangle =   \Psi(z_t,\dot z_t) + \Psi^*(z_t,D\S(z_t))\geq0.
\]

\begin{remark}
In a generalized gradient flow the role of the mobility is played by the (potentially nonlinear) map $\xi \mapsto D_\xi \Psi^*(z,\xi)$. Although this allows for nonlinear mobilities, the requirement that $D_\xi\Psi^*$ is a derivative means that not all nonlinear mobilities $M(z,\xi)\xi$ can be written as $M(z,\xi)\xi = D_\xi \Psi^*(z,\xi)$ for some $\Psi^*$; see~\cite{HutterSvendsen13} for a discussion.
\end{remark}





\section{Step III: Reversibility and large deviations lead to generalized gradient structures}

We now come to our main point. We generalize Onsager's result by connecting reversibility with \emph{generalized} gradient flows, but in a more general context than that of Theorem~\ref{th:Onsager}. Instead of a fixed system we choose a system with a parameter $n$, and consider the limit $n\to\infty$. This allows us to treat a very general class of systems using the theory of large deviations. 


Let $Z^{n}_t$ be a sequence of stochastic processes with transition kernel $P_t^{n}$ and invariant measure~$\mu^n$. Typically, $Z^n_t$ models a system with microscopic fluctuations, where $n$ is a large quantity like the number of particles in the system. Assume that 
\begin{enumerate}[label=(\Alph*)]
\item \label{LDP-cond1}
$P^n_t$ and $\mu^n$ satisfy the same reversibility condition as in Theorem~\ref{th:Onsager}, i.e. $\mu^n(dx_0)P^n_t(dx|x_0) = \mu^n(dx)P^n_t(dx_0|x)$; 
\item \label{LDP-cond2}
The equilibrium distributions $\mu^n$ satisfy a large-deviation principle with rate function~$\S$, i.e. if $Y^n$ are a random variables with distribution $\mu^n$, then
\[
  \Prob(Y^n\approx y) \mathop{\sim}_{n\to\infty} \exp\bigl[n \S(y)\bigr];
\]
\item \label{LDP-cond3}
For all $T>0$, $Z^{n}$ satisfies a large-deviation principle of the form
  \begin{align}
  \label{eq:path ldp general process}
    &\Prob_{\mu^n}\!\left(Z^{n}_t\big|_{t\in[0,T]} \approx z\big|_{t\in[0,T]}\right) \mathop{\sim}_{n\to\infty} \exp\bigl[n \left(\S(z_0)-I_T(z)\right)\bigr],
    &I_T(z):=\int_0^T\!\cL(z_t,\dot z_t)\,dt,
  \end{align}
  where the subscript $\mu^n$ indicates that the process starts at a position drawn from $\mu^n$.
\end{enumerate}

These three assumptions are natural generalizations of the conditions of Onsager's Theorem~\ref{th:Onsager}.
Condition~\ref{LDP-cond2} above is a natural generalization of the requirement that $\mu$ is Gaussian (Theorem~\ref{th:Onsager}, condition~\ref{th:Onsager:cond-mu-Gaussian}). Note that since probabilities are bounded by one, Condition~\ref{LDP-cond2} implies that $\S\leq0$, which can be seen as the large-deviation counterpart of the normalization of probability measures; the most probable behaviour in equilibrium is given by the maximizer(s) of $\S$ with value zero.

Condition~\ref{LDP-cond3} above is the natural extension of the assumption on the linear evolution of macroscopic quantities in Theorem~\ref{th:Onsager}, condition~\ref{th:Onsager:cond-evolution}, but this requires a little more explanation. The typical case is that $\cL$ is what we call an \emph{L-function}:
\begin{definition}We call~$\cL$ an L-function\footnote{The letter `L' in `L-function' refers to the word \emph{Lagrangian}, a name often used for expressions of the form of $I_T$ in~\eqref{eq:path ldp general process}. We thank the anonymous reviewer for pointing out that the name `L-function' has also been used to indicate \emph{Lyapunov} functions~(see e.g.~\cite{Beretta86}), which is different concept (but note that Theorem~\ref{th:main} implies that an L-function $\cL$ with certain properties \emph{generates} a Lyapunov function $\S$).}
 if 
\begin{enumerate}
  \item $\cL\geq0$,
  \item $\cL(z,\,\cdot\,)$ is convex for all $z$,
  \item $\cL$ induces an evolution equation $\dot z_t=\cA_\cL(z_t)$, in the sense that for all $(z,s)$,
\begin{equation}
\label{eq:L-function minimization}
  \cL(z,s)=0  \qquad \iff \qquad s = \cA_\cL(z).
\end{equation}
\end{enumerate}
\label{def:L-function}
\end{definition}
Again, the fact that $\min_s \cL(z,s)=0$ follows from the boundedness of probabilities.  The convexity is a standard consequence of the lower-semicontinuity given by large-deviation theory, but it will be important below.
The fact that $\cL(z,\cdot)$ has a single minimum, however, is a strong requirement, which is satisfied if in the limit $n\to\infty$ the behaviour of $Z^n_t$ becomes deterministic (a form of the law of large numbers). The limiting macroscopic equation is then $\dot z_t = \cA_\cL(z_t)$.

A very large class of systems satisfies the assumptions above: diffusion-type equations arising from stochastic lattice models, for instance~\cite{KipnisOllaVaradhan89}, or from many-particle Brownian motion~\cite{KipnisOlla90}, including convection and mean-field interaction~\cite{DawsonGartner87,FengKurtz06}. Chemical reactions are another important example, and we describe this in more detail in the next section. The method of Feng and Kurtz~\cite[Sec.~8.6.1.2]{FengKurtz06} provides an algorithm to derive such large-deviation principles in great generality.


The following theorem is our main result.
\begin{theorem}
\label{th:main}
Assume conditions \ref{LDP-cond1}--\ref{LDP-cond3} above. If $\cL$ is an L-function, then it generates a generalized gradient system $(\Z,\tfrac12\S,\Psi)$ with symmetric potential $\Psi$, i.e. the macroscopic evolution $\dot z_t = \cA_\cL(z_t)$ (see~\eqref{eq:L-function minimization}) can be written as the generalized gradient flow 
$\dot z_t = D\Psi^*\bigl(z_t,\tfrac12 D\S(z_t)\bigr)$. 
\end{theorem}

\begin{proof}
The proof consists of two steps:
\begin{enumerate}
\item \label{part1proof:rev}
The reversibility implies a symmetry relation for $\cL$ and $\S$,
\begin{equation}
\label{eq:symmetry-cL}
  \cL(z,s) - \cL(z,-s) = -\langle D\S(z),s\rangle.
\end{equation}
\item This symmetry implies that there exists a pair of potentials $\Psi,\Psi^*$, satisfying the requirements of Definition~\ref{def:gen gradient system}, such that 
\begin{equation}
\label{eq:LPsi-connection}
\cL(z,s) = \Psi(z,s) + \Psi^*\bigl(z,\tfrac12 D\S(z)\bigr) - \langle \tfrac12 D\S(z),s\rangle.
\end{equation}
\end{enumerate}

Part~\ref{part1proof:rev} is a well-known consequence of reversibility in the context of large deviations~\cite{Maes99,MaesRedigVan-Moffaert00,Seifert12}, and its proof runs as follows. For an arbitrary $(z,s)\in T\Z$, take any curve for which $(z_0,\dot z_0)=(z,s)$. From the reversibility follows the exact characterization, for all $T>0$,
\[
\Prob_{\mu^n}\!\left(Z^{n}_t\big|_{t\in[0,T]} \approx z\big|_{t\in[0,T]}\right)
= \Prob_{\mu^n}\!\left(Z^{n}_t\big|_{t\in[0,T]} \approx \sigma^{}_T z\big|_{t\in[0,T]}\right)
\]
where $\sigma^{}_T$ is the time-reversal operator $(\sigma^{}_T z)_t = z^{}_{T-t}$. By~\eqref{eq:path ldp general process} this implies
\begin{align*}
\S(z_0) - \int_0^T \cL(z_t,\dot z_t)\, dt 
  &= \S(z^{}_T) - \int_0^T \cL(z^{}_{T-t},-\dot z^{}_{T-t})\, dt \\
&= \S(z^{}_T) - \int_0^T \cL(z_s,-\dot z_s)\, ds .
\end{align*}
Upon dividing by $T$ and taking the limit $T\to0$ we find~\eqref{eq:symmetry-cL}.

The proof of part 2 follows from a manipulation of Legendre transforms; we refer to~\cite[Th.~2.1]{MielkePeletierRenger14} for the details.
\end{proof}

\section{Example: chemical reactions}

A large number of (generalized) gradient-flow formulations of macroscopic evolution equations can be connected to  large-deviation principles of their underlying stochastic processes, in the way that we describe above~\cite{AdamsDirrPeletierZimmer11,DuongLaschosRenger13,AdamsDirrPeletierZimmer13,DuongPeletierZimmer13,BonaschiPeletier14TR,MielkePeletierRenger14,LieroMielkePeletierRenger}. Here we discuss a single example, that of chemical reactions.

We consider chemical reactions of $I$ species $X_1,\dots,X_I$ of the form
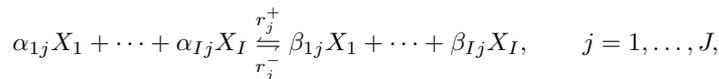
\[
\alpha_{1j}X_1 +\dots+\alpha_{Ij} X_I \overset{r_j^+}{\underset{r_j^-}\leftrightharpoons }
\beta_{1j}X_1 +\dots+\beta_{Ij} X_I, \qquad j=1,\dots, J,
\]
under constant temperature and pressure.

Two natural models of this system are the  Chemical Master Equation (e.g.~\cite{Gillespie00}) and the  Reaction-Rate Equation (e.g.~\cite{ErdiToth89}). The Chemical Master Equation (CME) describes a continuous-time, discrete-space stochastic process in terms of numbers of particles of each type; the Reaction-Rate Equation (RRE) is a deterministic equation for concentrations of species:
\begin{equation}
  \dot {\v c}_t = \sum_{j=1}^J \big(r_j^+(\v c_t) - r_j^-(\v c_t)\big)(\v\beta_j-\v\alpha_j).
\label{eq:RRE}
\end{equation}
Under reasonable conditions on the kinetics, in the limit of large volume and under well-mixed conditions, the CME becomes deterministic, and its limit is given by the RRE~\cite{Kurtz73}.


If the system satisfies both detailed balance and mass action, then~\eqref{eq:RRE} admits a stationary set of concentrations $\overline {\v c} = (\overline c_1,\dots,\overline c_I)$ and the net rate of reaction $j$ can be written as 
\[
r_j(\v c) := r_j^+(\v c)-r_j^-(\v c) = k_j  \bigl[{\v u}^{\v \alpha{j}} - 
   {\v u}^{\v \beta{j}} \bigr]
  \qquad \text{where}\qquad 
  u_i := \frac{c_i}{\overline c_i},
\]
and we write
\[
{\v u}^{\v \alpha{j}} := u_1^{\alpha_{1j}}\dots u_I^{\alpha_{Ij}}
\qquad\text{and}\quad 
{\v u}^{\v \beta{j}} := u_1^{\beta{1j}}\dots u_I^{\beta{Ij}}.
\]
Note that typically the evolution takes place in a submanifold, generated by the particular values of the quantities that the evolution preserves. Therefore $\overline c$ is not the only stationary state---each stoichiometic subspace, defined by a particular value of all conserved quantities, contains exactly one stationary state.


The generalized gradient structure $(\Z,\tfrac12\S,\Psi)$ for \eqref{eq:RRE} is given by the driving functional
\begin{equation}
\S(\v c) := -\sum_{i=1}^I \eta\Bigl(\frac{c_i}{\overline c_i}\Bigr) \overline c_i
  = -\sum_{i=1}^I \eta(u_i) \overline c_i, 
  \qquad \eta(s) := s \ln s - s + 1,
\label{eq:chemical free energy}
\end{equation}
and the exponential dissipation potential
\begin{equation}
\label{def:Psi-chemical-reactions}
\Psi_j^*(\v c,\v \xi) = 2k_j {\v u}^{\frac12\v\alpha_j+\frac12\v\beta_j}\big( \cosh\bigl((\v \alpha_j - \v \beta_j)\cdot \v \xi\bigr) -1 \big),
\qquad \Psi^*(\v c, \v \xi) = \sum_{j=1}^J \Psi_j^*(\v c,\v \xi).
\end{equation}
The functional $\S$ can be interpreted as a dimensionless negative free energy, for instance by noting that the derivative of $\S(\v c)$ with respect to a concentration $c_i$ equals
\[
\partial_{c_i} \S(\v c)  = -\eta'\Bigl(\frac{c_i}{\overline c_i}\Bigr) = -\log c_i + \log \overline c_i,
\]
which is equal, up to constant factors and a minus sign, to the common expression for the chemical potential $\mu_i = \mu_i^{\mathrm{std}} + RT\log c_i$.
Nonquadratic potentials such as~\eqref{def:Psi-chemical-reactions} have been discussed by Grmela~\cite{Grmela12} and similar expressions appear in the earlier literature~\cite{Feinberg72}.

The restriction to the submanifold mentioned above is implicitly implemented by $\Psi^*$: since $D\Psi^*_j$ is parallel to $\v \alpha_j - \v \beta_j$, the range of $D\Psi^*$ is the span of the vectors $\{\v \alpha_j - \v \beta_j\}_{j=1}^J$. 

For monomolecular reactions we prove in~\cite{MielkePeletierRenger14} that this generalized gradient system is indeed generated by the large deviations of the corresponding Chemical Master Equation, which is the discrete stochastic process in terms of numbers of particles of each species. In the forthcoming paper~\cite{MielkePattersonPeletierRenger15TA} we show that the same is true for general mass-action systems with detailed balance.  This shows that this generalized-gradient-flow structure of chemical reactions has its origins in the large-deviations behaviour of the underlying stochastic systems. The same is true for various other processes, such as diffusion~\cite{AdamsDirrPeletierZimmer13}, heat conduction~\cite{PeletierRedigVafayi14}, passive transport across membranes~\cite{GlitzkyMielke13}, and thermally activated motion~\cite{BonaschiPeletier14TR}, and we conjecture that this is actually a very general phenomenon.


\section{Discussion}

The strength of Onsager's theorem lies in the combination of its broad scope (all processes in a certain class) with a rigorous statement (the macroscopic equations have a certain structure). We generalize this theorem by enlargening the class to all sequences of processes with a large-deviation principle; our assertion, also rigorous, is an appropriate generalization of Onsager's symmetry statement, which coincides with the original symmetry in the case of linear processes.

Incidentally, this work shows how large-deviation theory, originally developed to  understand mathematically such physical concepts as free energy and entropy, provides us with the tools also to  understand dynamical generalizations of these concepts. In this way it appears to be the natural link between the stochastic microscopic description and the deterministic macroscopic one. 

Interestingly, the large-deviation viewpoint of this paper explains the appearance of nonquadratic dissipation potentials in thermodynamics: when the underlying stochastic process has non-Gaussian fluctuations, then the large-deviation rate functional has non-quadratic dependence on the generalized velocity, and the dissipation potential in the generalized gradient flow is non-quadratic. The non-quadratic nature therefore can be traced back to non-Gaussian fluctuations.

It would be very natural to consider systems with varying time-reversal parity, following Casimir~\cite{Casimir45}, but at this moment this is only understood for a single example~\cite{DuongPeletierZimmer13}.

\medskip

\textbf{Acknowledgement. } The authors are grateful to Gian Paolo Beretta, Miroslav Grmela, Markus H\"utter, and Hans-Christian \"Ottinger for their comments on an early version of this paper. AM and DRMR were partially supported by DFG via SFB 1114 (projects C5 and C8); MAP was supported by the NWO VICI grant 639.033.008.

%

\bibliographystyle{alphainitials}
\bibliography{IWNET-MielkePeletierRenger}
\end{document}